\documentclass{amsart}
\usepackage{amsaddr}
\usepackage{amssymb}
\usepackage{latexsym}
\usepackage{cmll}
\usepackage{stmaryrd}
\usepackage{eqnarray}
\usepackage{mathtools}
\usepackage{xspace}
\usepackage{algpseudocode}
\usepackage{multirow}
\usepackage{lineno}

\newtheorem{lm}{Lemma}[section]
\newtheorem{thm}[lm]{Theorem}

\newtheorem{cor}[lm]{Corollary}

\newtheorem{ex}[lm]{Example}

\newcounter{senumi}[section]
\newcounter{senumip}[section]
\newcounter{temp}[section]

\def\thesenumi{\thesection.\arabic{senumip}}
\def\p@senumip\thesenumip{\thesenumi}
    {\begin{list}%
        {(\thesenumi)}%
        {\usecounter{senumip}}
        \setcounter{senumip}{\value{temp}}
    }%
    {\setcounter{temp}{\value{senumip}}
     \end{list}}
\newcounter{penumi}[section]

\newcounter{ptemp}[section]

   {\begin{enumerate}%
        \setcounter{penumi}{\value{ptemp}}%
        \setcounter{enumi}{\value{ptemp}}%
    }%
    {\setcounter{ptemp}{\value{enumi}}
     \end{enumerate}}
\newcounter{ppenumi}[section]

\newcounter{pptemp}[section]

\def\theppenumi{\theptemp.\arabic{ppenumi}}
    {\begin{list}%
        {(\theppenumi)}%
        {\usecounter{ppenumi}\setlength{\rightmargin}{\leftmargin}}
        \setcounter{ppenumi}{\value{pptemp}}
    }%
    {\setcounter{pptemp}{\value{ppenumi}}
     \end{list}}

\newenvironment{tenumerate}{\begin{enumerate}}{\end{enumerate}}

\newcommand{\m}[1]{{\uppercase {\mathbf{#1}}}}

\newcommand{\ceqv}[1]{\ensuremath{\operatorname{\textsc{\upshape{Ceqv}}}
                                \ifthenelse{\equal{#1}{}}{}{\!\left( { #1} \right)}}}
\newcommand{\csat}[1]{\ensuremath{\operatorname{\textsc{\upshape{Csat}}}
                                \ifthenelse{\equal{#1}{}}{}{\!\left( { #1} \right)}}}
\newcommand{\Csat}[1]{\ensuremath{\operatorname{\textsc{\upshape{Csat}}}
                                \ifthenelse{\equal{#1}{}}{}{\!\left( { #1} \right)}}}
\newcommand{\Ceqv}[1]{\ensuremath{\operatorname{\textsc{\upshape{Ceqv}}}
                                \ifthenelse{\equal{#1}{}}{}{\!\left( { #1} \right)}}}
\newcommand{\mcsat}[1]{\ensuremath{\operatorname{\textsc{\upshape{MCsat}}}
                                \ifthenelse{\equal{#1}{}}{}{\!\left( {\m #1} \right)}}}
\newcommand{\SCsat}[1]{\ensuremath{\operatorname{\textsc{\upshape{SCsat}}}
                                \ifthenelse{\equal{#1}{}}{}{\!\left( {\m #1} \right)}}}

\newcommand{\csp}[1]{\ensuremath{\operatorname{\textsc{\upshape{CSP}}}
                                \ifthenelse{\equal{#1}{}}{}{\!\left( {#1} \right)}}}

\newcommand{\npc}{\textsf{NP}-complete\xspace}

\newcommand{\np}{\textsf{NP}\xspace}
\newcommand{\ptime}{\textsf{P}\xspace}

\newcommand{\set}[1]{{\left\{ {#1} \right\} }}

\renewcommand{\leq}{\leqslant}
\renewcommand{\geq}{\geqslant}

\renewcommand{\mapsto}{\longmapsto}

\newcommand {\bc}[1]{{\overline {#1}} }

\newcommand{\con}[1]{{\sf Con\:\m{#1}}}
\newcommand{\cn}[1]{{\sf Con\:\m{#1}}}

\newcommand{\po}[1]{{\mathbf {#1}}}
\newcommand{\te}[1]{{\mathbf {#1}}}

\renewcommand{\o}[1]{\overline {#1}}
\newcounter{ttable}

\newcommand{\comm}[2]{\left[ #1 , #2 \right]}
\newcommand{\com}[2]{\left[ #1 , #2 \right]}
\newcommand{\commm}[2]{\left[ #1 , \ldots, #2 \right]}

\newcommand{\map}{\longrightarrow}

\newcommand{\congruent}[1]{\stackrel{#1}{\equiv}}

\author{Piotr Kawa\l{}ek}
\address{Jagiellonian University, Faculty of Mathematics and Computer Science, Department of Theoretical Computer Science ul.~Prof.~S.~\L{}ojasiewicza~6,3`0-348,~Krak\'ow,~Poland }
\email{piotr.kawalek@doctoral.tcs.uj.edu.pl}{}

\author{Jacek Krzaczkowski}
\address{Maria Curie-Sklodowska University, Faculty of Mathematics, Physics and Computer Science, Department of Computer Science ul.~Akademicka~9,~20-033,~Lublin,~Poland }
\email{krzacz@poczta.umcs.lublin.pl}

\keywords{circuit satisfiability, solving equations, supernilpotent algebras, satisfiability in groups}

\begin{document}

\thanks{The first author was partially supported by Polish NCN Grant \# 2014/14/A/ST6/00138.}
\begin{abstract}
 In this paper two algorithms solving circuit satisfiability problem over supernilpotent algebras are presented. The first one is deterministic and is faster than fastest previous algorithm presented in \cite{aichinger}. The second one is probabilistic with linear time complexity. Application of the former algorithm to finite groups provides time complexity that is usually lower than in previously best \cite{foldvari2018} and application of the latter leads to corollary, that circuit satisfiability problem for group $\m G$ is either tractable in probabilistic linear time if $\m G$ is nilpotent or is \npc if $\m G$ fails to be nilpotent. The results are obtained, by translating  equations between polynomials over supernilpotent algebras to bounded degree polynomial equations over finite fields.
\end{abstract}

\title{Even faster algorithms for CSAT over~supernilpotent algebras}
\maketitle

\section{Introduction}
Solving equations is one of the most popular mathematical problems with applications in many areas. We are interested in computational complexity of equations satisfiability problem for fixed finite algebra. In the original definition of the problem for a given equation of polynomials over fixed algebra we ask if it has solution or not.  There is number of papers in which authors tried to characterize algebras for which this problem is tractable in polynomial time and for which it is hard in terms of some well established complexity assumptions (i. e. \ptime $\neq$ \np). Most of authors consider some well known structures with fixed language like groups \cite{goldman-russell}, \cite{horvath-szabo:groups}, \cite{horvath:positive}, \cite{horvath:metabelian}, \cite{foldvari2018}, \cite{foldvari-horvath}, rings \cite{horvath:positive}, \cite{karolyi-szabo} or lattices \cite{schwarz}, however there is some number of papers considering more general cases e.g. \cite{gorazd-krzacz:2el}, \cite{gorazd-krzacz:preprimal}, \cite{aichinger}. A new look on the problem was proposed in
 \cite{ik:lics18}. This paper was the first systematic study on solving equations in quite general setting. The authors of \cite{ik:lics18} decided to allow more compact representation of polynomials on the input of the problem, so to represent them as multi-valued circuits. 
 It leads to the following definition of the problem:
 \begin{description}
 \item[\csat{\m A}] given a circuit over $\m A$ with two output gates $g_1$, $g_2$ is there a valuation  of  input  gates $\o x$ that  gives  the  same  output  on $g_1$, $g_2$,  i.e. $g_1(\o x) =g_2(\o x)$.
 \end{description}
 Such the definition  gives us, that computational complexity of \csat{\m A} depends only on the polynomial clone of $\m A$  and in consequence can be characterized in terms of algebraic properties of $\m A$.
 Several articles considering this new approach to solving equations have appeared e.g. \cite{komp:2018}, \cite{ikk:mfcs18}, \cite{aichinger}, \cite{kkk:2nil}, \cite{komp:cc}, \cite{ikk:lics20}. In this paper we will present the results in terms of $\csat{}$, however for clarity we mention, that all the algorithms and upper bounds presented here apply also to the original definition of the problem as polynomials can be represented by circuits expanding size of the representation only by constant factor. 

Algebras generating  congruence modular variety are the wide class of algebras containing among others many popular algebraic structures like groups, rings and lattices.  We will call this class of algebras $\textsc{CM}$  for short. Analyzing partial characterization of computational complexity of \csat{} for algebras from $\textsc{CM}$ presented in \cite{ik:lics18} and also results of \cite{ikk:mfcs18}, \cite{komp:cc} and \cite{ikk:lics20}, we can see the truly rich world in which one can find problems of different complexities: \npc problems, problems contained in \ptime and those, that are natural candidates for \np-intermediate problems.  Surprisingly, there are known only three essentially different polynomial time algorithms solving \csat{} over algebras from congruence modular varieties.  Two of them are the  black-box algorithms i.e. the algorithms which treat circuits as a black-box and try to find the solution by checking not too big set of potential solutions (so called hitting set). One of them originally was proposed for nilpotent groups \cite{goldman-russell} and the second one works for distributive lattices (\cite{schwarz}). The third of the algorithms mentioned above solves \csat{} by inspecting some kind of normal form of a given circuit but it seems that the usefulness of such kind of  algorithm is limited to so called $2$-step supernilpotent algebras \cite{ikk:mfcs18}, \cite{ikk:lics20}.

In this paper we consider supernilpotent algebras from $\textsc{CM}$ which are natural generalization of nilpotent groups (among all groups only those nilpotent ones induce tractable problems, assuming $\ptime \neq \np$). Every such supernilpotent algebra $\m A$ decomposes into a direct product of supernilpotent algebras of prime power order. That is why we can reduce problem of solving equations over $\m A$ to fixed number (at most $\log |A|$) of satisfiability problems over supernilpotent algebras, but this time of prime power order. This Turing reduction can be performed in linear time and thus we will only be looking for an algorithm for solving equations over supernilpotent algebras of prime power order.

We will slightly modify algorithm that was applied in the group setting. In this algorithm we check potential solutions in which at most $d$ variables are assigned to non-zero value. It was introduced by Goldmann and Russell in \cite{goldman-russell} for nilpotent groups and its correctness was reproved by Horvath  in \cite{horvath:positive}. In both cases it was shown that considered algorithm works in polynomial time but the degree of the polynomial came from application of Ramsey Theory and was really huge. Later it was independently shown in \cite{komp:2018} and \cite{ik:lics18} that essentially the same algorithm works for supernilpotent algebras from congruence modular variety in polynomial time  with the same huge degree of the polynomial. This results was improved by Aichinger in \cite{aichinger}. In his paper the degree of the polynomial describing complexity of \csat{\m A} was bounded by $d=|A|^{\log_2{|A|}+\log_2{m}+1}$, where $m$ is a maximal arity of basic operation of $\m A$.  Using similar tools as Aichinger and some new ideas we show the following. 

\begin{thm}\label{thm-determ-alg}
Let $\m A$ be a supernilpotent algebra of prime power order $q^h$ from congruence modular variety. Then there exists black-box algorithm solving \csat{\m A} in time $O(n^d k)$, where $d=|A|^{\log_q{m}+1}$, $m$ is a maximal arity of basic operations of $\m A$ and $k$ is the input size.
\end{thm}

Proof of this theorem can be found in Section \ref{sec-determ}. Note, that after applying Theorem \ref{thm-determ-alg} in the realm of nilpotent groups of prime power order $q^h$ we obtain that $d=|G|^{\log_q{2}+1}$. We note here that in \cite{foldvari2018} A. Földvári, using some group specific tools, showed the different algorithm  for original equation satisfiability problem of  time complexity $O(n^d)$, where $d=\frac{1}{ 2}|G|^2\cdot \log{|G|}$ (here $n$ denotes the input size).  So in most cases our algorithm improves this result too (especially when prime $q$ is huge).

It turns out that switching from deterministic computational model to probabilistic one we obtain a great improvement. It is shown in the second main result of this paper, which states the following

\begin{thm}\label{thm-rand}
Let $\m A$ be a supernilpotent algebra of prime power order from congruence modular variety.  Then there exists linear time Monte Carlo algorithm solving \csat{\m A}.
\end{thm}

The surprising corollary we get when we apply Theorem \ref{thm-rand} to finite groups and use results from \cite{goldman-russell} and \cite{horvath-szabo:polsatstar} 
\begin{cor}
Let $\m G$ be a finite group. Then \csat{\m G}
\begin{itemize}
 \item can be solved by linear time Monte Carlo algorithm if $\m G$ is nilpotent,
  \item is \npc otherwise.
\end{itemize}
\end{cor}

To obtain algorithms mentioned above we study structure of nilpotent algebras of prime power order. Thanks to deep universal algebraic tools developed in \cite{fm} and \cite{wanderwerf} our study does not contain hard to read technical proofs. Nevertheless reader not interested in algebraic details can skip Section \ref{sec-struct}. Reader interested in more systematic and detailed study in this spirit but in more general settings can see \cite{ikk:intervals}.

The main conclusion of Section \ref{sec-struct} is that solving equations over nilpotent algebras of prime power order $q^h$ can be reduced to solving one special equation between polynomials over field $\m F_q$ of bounded degree. Thus, in next sections we do not need the universal algebraic tools and we work with finite fields only.

Our randomized algorithm solving equations over $\m F_q$ of low degree is very simple.  It turned out that all we need to do is randomly draw solutions with an uniform distribution. In such the way, we obtain $c$-correct true-biased algorithm for some constant $c$ depending on the algebra. It works thanks to nice behavior of polynomial over $\m F_q$ of not too high degree.  This behavior is described in the following Lemma.

\begin{lm}\label{lm-zippel}
Let $\po f$ be $n$-ary polynomial of degree $d$ over finite field $\m F_q$. Then, for every $y\in F_q$ such that $|\po f^{-1}(y)|>0$ we have $|\po f^{-1}(y)|\geq q^{n-d -q\log_2{q}}$
\end{lm}
Note that if degree of polynomial was smaller than the size of the field, then we would just need to apply famous Schwartz–Zippel lemma to get that density of solutions among all possible assignments to variables  is huge. In our case degree of polynomial is bounded by constant depending on $\m A$ and almost always it exceeds the field size we are working with. There are also number of another results, that can be applied here, introduced for polynomial identity checking of $s$-sparse polynomials, but they do not lead to linear time algorithm.

The article is organized as follows. The second section contains some definitions and  background materials. In Section \ref{sec-struct} we present the structure of supernilpotent algebras and show that \csat{} for such algebras can be reduced to solving equations between polynomials of bounded degree over finite field.  The proof of Lemma \ref{lm-zippel} is contained in Section \ref{sec-zippel}.
In Sections \ref{sec-determ} and  \ref{sec-rand} we show deterministic and randomized algorithms solving \csat{} for supernilpotent algebras and prove Theorem \ref{thm-determ-alg} and Theorem \ref{thm-rand}. Finally, Section \ref{sec-conc} contains remarks regarding results contained in this paper and conclusions.
\section{Background material}
In this paper we use the standard notation of universal algebra (see e.g. \cite{burris}). An algebra is for us a structure consisting of  the set called universe and the set of finitary operations on it. Groups and fields are obviously examples of algebras.  All algebras considered in this paper are finite i.e with finite universe and finite set of operations. We usually denote algebras using bold capital letter and its universe by the same but non-bold letter. The language or type of algebra is the set $\mathcal{F}$ of function symbols together with non-negative integers assigned to each member of $\mathcal{F}$. We say that an algebra $\m A=(A,F)$ is of type $\mathcal{F}$ if the set $F$ of its operations is indexed by elements of $\mathcal{F}$ and for $n$-ary function symbol the corresponding operation $f^{\m A}\in F$ is also $n$-ary. We use overlined small letters e.g. $\o x$, $\o a$ to denote tuples of variables or elements of an algebra and the same letters without overline but with subscript to denote elements of tuples e.g. $x_i$, $a_i$.

Now, we will recall some basic notions. Let $\m A$ be an algebra and $\alpha, \beta, \gamma \in \cn A$.
We say that {\em $\alpha$ centralizes $\beta$ modulo $\gamma$}, denoted $C(\alpha, \beta;\gamma)$, if for every $n$ and $n$-ary term $\te t$, every $(a,b) \in \alpha$, and every
$(c_1,d_1),\dots,(c_n,d_n)\in \alpha$ we have
\[
\te t(a,\bc c) \congruent{\gamma} \te t(a,\bc d)  \mbox{\ \  iff \ \  }
\te t(b,\bc c) \congruent{\gamma} \te t(b,\bc d).
\]

If $\alpha$ and $\beta$ are congruence relations on an algebra
$\m A$, then the {\em commutator} of $\alpha$ and $\beta$, denoted
$\com \alpha \beta$, is the least congruence $\gamma$ for which
$C(\alpha,\beta;\gamma)$. Note that for algebras from congruence modular variety defined in this manner commutator is commutative, monotone  and join-distributive. We say that $\alpha$ is {\em abelian over} $\beta$ if $[\alpha,\alpha]\leq\beta$. An algebra $\m A$ is {\em abelian} if $[1_{\m A},1_{\m A}]=0_{\m A}$. Note, that in \textsc{CM} abelian algebras are exactly affine algebras i.e. algebras polynomially equivalent to a module.

For a congruence $\theta$ and $i=1,2,\dots$ we write
\[ \begin{array}{rclcrcl}
\theta^{(1)}&=&\theta \\
\theta^{(i+1)}&=&[\theta,\theta^{(i)}] 
\end{array}
\]

A congruence relation $\theta$ on $\m A$ is called
{\em $k$-step nilpotent}
if $\theta^{(k+1)}=0_A$ and the algebra $\m A$
is {\em nilpotent} if $1_A$ is $k$-step nilpotent for some finite $k$.

Our study is focused on supernilpotency - the strengthening of the nilpotency.
For congruences
$\alpha_1,\ldots,\alpha_k,\beta,\gamma \in \con A$
we say that $\alpha_1,\ldots,\alpha_k$ centralize $\beta$ modulo $\gamma$,
and write $C(\alpha_1,\ldots,\alpha_k,\beta;\gamma)$,
if for every polynomial $\po f$ over $\m A$ and all tuples
$\o a_1 \congruent{\alpha_1} \o b_1, \ldots, \o a_k \congruent{\alpha_k} \o b_k$
and $\o u \congruent{\beta} \o v$
such that
\[
\po f(\o x_1,\ldots, \o x_k, \o u) \congruent{\gamma} \po f(\o x_1,\ldots, \o x_k, \o v)
\]
for all possible choices of
$(\o x_1,\ldots, \o x_k)$ in $\set{\o a_1,\o b_1} \times \ldots \times \set{\o a_k,\o b_k}$
but $(\o b_1,\ldots.\o b_k)$,
we also have
\[
\po f(\o b_1,\ldots, \o b_k, \o u) \congruent{\gamma} \po f(\o b_1,\ldots, \o b_k, \o v).
\]
This notion was introduced by A.~Bulatov \cite{bulatov:supercomm}
and further developed by E.~Aichinger and N.~Mudrinski
\cite{aichinger-mudrinski}.
In particular they have shown that for all $\alpha_1,\ldots,\alpha_k \in \con A$
there is the smallest congruence $\gamma$ with $C(\alpha_1,\ldots,\alpha_k;\gamma)$
called the $k$-ary commutator and denoted by $\commm{\alpha_1}{\alpha_k}$.
Such generalized commutator for algebras from congruence modular varieties. has many nice properties. In particular this commutator is symmetric, monotone, join-distributive and we have
\begin{equation}\label{eq-spercomm}
\comm{\alpha_1} {\commm{\alpha_2}{\alpha_k}} \leq \commm{\alpha_1}{\alpha_k}\leq\commm{\alpha_2}{\alpha_k}
\end{equation}
Generalization of commutator enabled us to define  $k$-supernilpotent algebras as algebras satisfying
$[ \overbrace{1,\ldots,1}^{\text{\scriptsize $k\!+\!1$ times}} ] =0$. An algebra is called supernilpotent if it is $k$-supernilpotent for some $k$.
Note that by \eqref{eq-spercomm} every $k$-supernilpotent algebra from congruence modular variety is $k$-nilpotent.
Moreover supernilpotent algebras from congruence modular variety have very nice characterization which can be easily inferred from
the deep work of R.~Freese and R.~McKenzie \cite{fm}
and K.~Kearnes \cite{kearnes:small-free}, and have been observed in \cite{aichinger-mudrinski}.

\begin{thm}
\label{thm-prime-supernil}
For a finite algebra $\m A$ 
from a congruence modular variety the following conditions are equivalent:
\begin{tenumerate}
\item
$\m A$ is $k$-supernilpotent,
\item
$\m A$ is $k$-nilpotent, decomposes into a direct product of algebras of prime power order
and the term clone of $\m A$ is generated by finitely many operations,
\item
$\m A$ is $k$-nilpotent and all commutator polynomials have rank at most $k$.
\end{tenumerate}
\end{thm}

We will see in the next sections that Theorem \ref{thm-prime-supernil} shows two key properties of supernilpotent algebras: possibility of decomposition into direct product of algebras of prime power order and bounded essential arity of commutator polynomials. The second property can be formulate in a less formal way that for every $k$-supernilpotent algebra there is no possibility to express as a polynomial a function which behave similarly to $k+1$-ary conjunction. 

\section{The structure of supernilpotent algebras}\label{sec-struct}
 In this section we will see that every supernilpotent algebra of prime power order $q^h$ is in fact a wreath product of algebras polynomially equivalent to  simple modules of order $q^\alpha$. In fact, we will see even more, we will prove that every operation of such the algebra can be described by a bunch of polynomials over $\m F_q$ of bounded degree.  Then using this characterization we will able to show some facts needed in the next sections. More detailed investigations of structure of supernilpotent and not only supernilpotent algebras can be find in \cite{ikk:intervals}.

First, we present mentioned earlier decomposition of supernilpotent algebras into wreath product of  algebras polynomially equivalent with simple abelian groups. We will use Freese's and McKenzie's ideas from \cite{fm} developed in more general settings in WanderWerf's PhD thesis \cite{wanderwerf}. In particular for algebras $\m Q=(Q,F^{\m Q})$ and $\m B=(B,F^{\m B})$ of the same type $\mathcal{F}$ such that $\m Q$ is abelian with associated group $(Q,+,-)$ and the set of operation $\m T$ such that for $n$-ary operation $f\in \mathcal{F}$ there is $t_f:B^{n}\mapsto Q$ in \cite{fm} was defined algebra $\m A=\m Q\otimes^{T}\m B$ of type $\mathcal{F}$ with universe $Q\times B$ and operations defined as follow 
\[
f^{\m A}((q_1,b_1),\ldots,(q_n,b_n))=\left(f^{\m Q}(q_1,\ldots,q_n)+t(b_1,\ldots,b_n),f^{\m B}(b_1,\ldots,b_n)\right),
\]
where $f$ is $n$-ary operation from $\mathcal{F}$. Note that since $\m Q$ is an abelian algebra from congruence modular variety and hence affine $f^{\m Q}$ can be expressed  in the form $f^{\m Q}(q_1,\ldots, g_n)=\sum_{i=1}^{n}\lambda_i q_i+c$, where $\lambda_i$'s are endomorphisms of $(Q,+)$. 

Let assume that $\m A$ is supernilpotent algebra of prime power order $q^{h}$ and $\theta\in\con{\m A}$ be one of its atoms (i.e. conqruences covering $0_{\m A}$). Then using results form \cite{fm} it can be shown that $\m A$ can be decomposed into wreath product of $\m A/\theta$ and some algebra $\m Q$ polynomially equivalent to simple module. More precisely $\m A$ is isomorphic to the algebra $\m Q\otimes^{T}\m A/\theta$ for some $T$ and $\m Q$. Note that if $|Q| = q^{\alpha}$ then $\m A/\theta$ has order $q^{h-\alpha}$. Repeating this procedure recursively for $\m A/\theta$ we obtain that $\m A$ is isomorphic to some algebra which is the wreath product of algebras polynomially equivalent to simple modules of order $p^{\alpha_1}. \ldots, p^{\alpha_s}$. From this point we assume that $\m A$ itself is such the algebra. Denote $e_i$ the projection on the $i$-th coordinate of $A$ (for $i=1\ldots s)$. Now enrolling the recursive procedure we get, that every basic operation $f$ of $\m A$ fulfills the following properties
\begin{gather*}
e_s(f(x_1,\ldots,x_n))=\sum_{i=1}^{n}\lambda_i^s e_i(x_i)+t^s_f,\\
\ldots\\
e_j(f(x_1,\ldots,x_n))=\sum_{i=1}^{n}\lambda_i^j e_j(x_i)+t^j_f(e_{j+1}(x_1),\ldots,e_{s}(x_1),\ldots,e_{j+1}( x_n),\ldots,e_{s}(x_n)),
\end{gather*}
for some $\lambda_i^j$'s being endomorphisms of $j$-th module (of order $p^{\alpha_j}$) and some $t_f^j$'s. 
Note that constant summands in above expressions are hidden in $t_f^j$'s and  $t^s_f$ is just the constant. 

We will now translate every polynomial $\po g$ over $\m A$ to system of polynomials over the field $\m F_q$ that will simulate the behaviour of $\po g$, From the above observations about wreath product we see, that every element $a \in A$ can be written as a tuple $a = (e_1 a. \ldots, e_s a)$. Furthermore each $e_i a$ can be identified with a tuple $b_1, \ldots, b_{\alpha_s}$ where each $b_j \in Z_q$ . Indeed, each simple module of size $q^{\alpha}$ has a group underlay of prime exponent, this group must be then isomorphic to group $\m Z_q^{\alpha}$. So each element $a \in A$ can be identified in such a way with tuple $(\pi_1(a), \ldots, \pi_h (a))$ (with $\pi_i(a) \in Z_q$) and without loss of generality we will just write $a = (a_1, \ldots, a_h)$ (as we can replace algebra $\m A$ with isomorphic algebra accordingly) or $a=(a_1, a_2. \ldots, a_{\alpha_i})$ when $a \in e_i A$. 

So now it is clear, that for $i=1\ldots h$ each $\pi_i \po g(x_1, \ldots, x_n)$ is in fact the function from $(Z_q)^{nh} \map Z_q$ so as such can be represented by multivariate polynomial over variables $\pi_1 x_1, \ldots \pi_h x_1,\ldots,\pi_1 x_n, \ldots \pi_h x_n$. So for each $i=1\ldots h$ we have some polynomial $\po p_i$ satisfying $\pi_i \po g(x_1, \ldots, x_n) = \po p_i(\pi_1 x_1, \ldots \pi_h x_1,\ldots,\pi_1 x_n, \ldots \pi_h x_n)$. We know from basic algebra that $\po p_i$ has unique representation up to equations $x^{q} = x$ (for all variables). We will always mean by polynomial representing $\pi_i \po g$ this of smallest total degree up to those equations. We will also write $\deg \pi_i \po g$ for the degree of polynomial representing $\pi_i \po g$. We now want to prove, that such polynomials have small degrees.
 
\begin{lm}
\label{sys-pol}
Let $\m A$ be supernilpotent algebra of prime power order $q^h$ and $\po g$ be $n$-ary polynomial of $\m A$. Let $d_i$ be maximal degree of $\pi_j \po g$ for $\alpha_{1}+\ldots+\alpha_{i-1}< j\leq \alpha_1+\ldots+\alpha_{i-1}+\alpha_{i}$. Then
$$\sum_{i=1}^s \alpha_i \cdot d_i \leq (mq)^{\alpha_1 + \ldots + \alpha_{s-1}} \cdot \alpha_s$$
where $m$ is maximal arity of basic operation in the signature of $\m A$.
\end{lm}
\begin{proof}
We will inductively decrease $j=s \ldots 1$ and consider coordinates of $e_j A$ (there is $\alpha_j$ of them) to obtain degree of $\pi_j \po g$ for $\alpha_{1}+\ldots+\alpha_{i-1}< j\leq \alpha_1+\ldots+\alpha_{i-1}+\alpha_{i}$. Observe, that from the form of any basic operation of $\m A$ that we unrolled from wreath product representation  we can get (by simple induction) that for any $n$-ary polynomial $\po g$ its $j$-th coordinate $e_j \po g$ can be written as sum of elements of one of the forms:
\begin{itemize}
    \item $\lambda e_j x_i$, where $x_i$ is variable and $\lambda$ is some endomorphism of module corresponding to $e_j \m A$,
    \item $t_f^j (e_{j+1}\po g^{(1)},\ldots,  e_{s}\po  g^{(1)} \ldots, \po e_{j+1}\po g^{(l)}\ldots, e_{s} \po g^{(l)})$, where $t_f^j$ comes from $l$-ary basic operation $f$ of the algebra $\m A$ and $\po g^{(i)}$ are other polynomials of $\m A$,
    \item constant,
\end{itemize}
and for $j=s$ we do not have the second type of the above summands. To start with take $j=s$. $e_s A$ is then underlying set of a module of size $q^{\alpha_s}$ so it has $\alpha_s$ coordinates. We want to bound degree of polynomial representing $e_s f$ projected to each such coordinate. Notice that $\lambda e_s x_i$ is essentially an unary function, that depends only on projections of $x_i$ to $\alpha_s$ coordinates. Moreover on each coordinate it must be a linear function, because $\lambda$ is endomorphism of abelian group of exponent q. It means that on each coordinate it can be represented by polynomial of degree at most $1$. So we get that $d_s \leq 1$ (because degree of sum of polynomials is at most maximal degree of the summand and adding constants does not affect our upper bound).

In case $j<s$ we again bound degrees of polynomials for $\lambda e_j x_i$ by $1$ and we are left with the summands of the form $t_f^j (e_{j+1}\po g^{(1)},\ldots,e_{s}\po g^{(1)},\ldots,e_{j+1}\po g^{(l)},\ldots, e_{s}\po g^{(l)})$, where $l$ is arity of basic operation $f$. 
For $u>j$ each $e_{u}\po g^{(v)}$ can be represented by $\alpha_u$ polynomials of degree at most $d_u$. As every projection of $t_f^j$ itself can be represented as polynomial whose each of variable appears with degree at most $q-1$, so $t_f^j (e_{j+1}\po g^{(1)},\ldots,e_{s}\po g^{(1)},\ldots,e_{j+1}\po g^{(l)},\ldots, e_{s} \po g^{(l)})$ projected to any of its $\alpha_j$ coordinates can be represented by polynomial of degree 

$$d_j \leq l\cdot (q-1) \cdot \sum_{i=j+1}^s \alpha_i d_i$$
Since it works for any $j$ we have that:

$$\sum_{i=1}^s \alpha_i d_i = \alpha_1 d_1 + \sum_{i=2}^s \alpha_i d_i  \leq 
\alpha_1 \cdot l\cdot (q-1) \cdot \sum_{i=2}^s \alpha_i d_i +  \sum_{i=2}^s \alpha_i d_i = 
((q-1)l\alpha_1 + 1)(\sum_{i=2}^s \alpha_i d_i)$$ 
As $(q-1)l\alpha_1 + 1 \leq (ql)^{\alpha_1}$ we get

$$\sum_{i=1}^s \alpha_i d_i \leq (ql)^{\alpha_1} \cdot (\sum_{i=2}^s \alpha_i d_i),$$
and applying the same reasoning recursively to $\sum_{i=j}^s \alpha_i d_i$ for $j=2,3,\ldots,s$ we will end up

$$\sum_{i=1}^s \alpha_i d_i \leq (ql)^{\alpha_1} (ql)^{\alpha_2} \cdot (ql)^{\alpha_{s-1}} \alpha_s d_s= (ql)^{\alpha_1 + \ldots + \alpha_{s-1}} \cdot \alpha_s$$
what we wanted to prove.

\end{proof}

Lemma  \ref{sys-pol} shows in fact, how to reduce solving equations over supernilpotent algebra $\m A$ of prime power order $q^h$ to system of $h$ equations over field $\m F_q$.  Now, we would like to reduce solving equations over $\m A$ to solving one equations of the from $\po p(\overline{x}) = 1$, where $\po p$ is bounded degree polynomial over field $\m F_q$.  Moreover, the lemma shows that there is easy to compute one to one mapping between solutions of new equation and the original one. 

\begin{lm}\label{lm-eqA-to-eqF}
Let $\m  A$ be supernilpotent algebra of prime power order $q^h$. Then for $n$-ary $\po p$ and $\po g$ polynomials over $\m A$ there exists $nh$-ary polynomial $\po f$ over $\m F_q$ of degree at most $|A|^{\log_q m + 1}$ such that $f(F^{hn}_q)\subseteq\set{0,1}$ and for $\o a\in A^n $
$$\po p(a_1, \ldots, a_n) = \po g(a_1, \ldots, a_n)$$ 
iff
$$\po f(\pi_1 a_1, \ldots,\pi_h a_1,\ldots,\pi_1 a_n,\ldots,\pi_h x_n) = 1.$$
\end{lm}
\begin{proof}
Let 
\begin{equation}\label{eq-reducing}
\po p(x_1,\ldots,x_n)=\po g(x_1,\ldots,x_n)
\end{equation}
be an equation over $\m A$. Note that every polynomial of $\m A$ projected by every $\pi_i$ can be represented by polynomial over field $\m F_q$. So naturally we can write our equations equivalently as system of $h$ polynomial equations:
 
\begin{equation}\label{eq-system}
\begin{cases}
 \po p_1(\pi_1 x_1, \ldots, \pi_h x_n)=0\\
 \po p_2(\pi_1 x_1, \ldots, \pi_h x_n)=0\\
 \ldots\\
 \po p_h(\pi_1 x_1, \ldots, \pi_h x_n)=0
\end{cases}
\end{equation}

 It is easy to see that function defined as follows
\begin{equation}\label{eq-prod}
\po f(\o x)=\prod_{i=1}^{h}(1-\po p_i(\o x)^{q-1})
\end{equation}
fulfills conditions of the Lemma. It left to count degree of $\po f$. As $\alpha_j$ of those polynomials have degree bounded by $d_j$ for $j=1 \ldots s$ we get that degree of $\po f$ is bounded by $(q-1)(\sum_{i=1}^s \alpha_i d_i)$ So by lemma \ref{sys-pol} as $q^{\alpha_1 + \ldots + \alpha_n} = |A|$ this is bounded by 
$$(q-1)\cdot(mq)^{\alpha_1 + \ldots + \alpha_{s-1}} \cdot \alpha_s \leq (mq)^{\alpha_1 + \ldots, + \alpha_s} = |A|^{\log_q m + 1} $$
\end{proof}

\section{Behavior of polynomials over finite fields}\label{sec-zippel}
This section contains proof of Lemma \ref{lm-zippel}. The main idea of the proof is to show that given polynomial over a finite field can be transformed into some special polynomial of known degree. The way we do this transformation allow us to establish the lower bound of the given polynomial's degree depending among other on the inverse image of chosen element of the field. Hence, by elementary calculations we obtain that the statement of the lemma holds.  

Let $\po f$ be a $n$-ary polynomial over field $\m F_q$ for some prime $q$. We will prove that for every $y\in \po f(F_q)$ we have that $|\po f^{-1}(y)|>q^{n-\deg \po f - q\log_2{q}}$.  Since for a constant polynomial this is obviously true, we assume that $\po f$ is not constant. Fix  $y\in \po f(F_q)$. We will construct the sequence of at most $n$ polynomials of decreasing arity such that:
\begin{itemize}
 \item $\po f_0=\po f$,
 \item arity of $\po f_i$ is $n-i$,
 \item $\frac{|\po f_i^{-1}(y)|}{c}\geq|\po f_{i+1}^{-1}(y)|>0$, where $c\in\set{2,q}$,
 \item polynomial $\po f_{i+1}$ is obtained by substituting some variable in $\po f_i$ by constant or linear combination of other variables,
 \item if $\po f_l$ is the last polynomial in the sequence then either $|\po f_{l}^{-1}(y)|=1$ or $\po f_l$ is a polynomial in one variable.
\end{itemize}

We start with defining the sequence $\set{\po f_i}_{i=0}^{l}$. Let $\po f_0=\po f$. If arity of $\po f_i$ is higher then $1$ and $|\po f_i^{-1}(y)|>1$ then we define $\po f_{i+1}$ in one of two ways depending on the size of $|\po f_i^{-1}(y)|>1$. If $1<|\po f_i^{-1}(y)|<q^q$ then there exists $\o a, \o b\in \po f_i^{-1}(y)$ such that 
$\o a \neq \o b$. Since $\o a$ and $\o b$ are not equal we can choose $j$ such that $a_j\neq b_j$. Without loss of generality assume $j = n-i$.  Now we obtain $\po f_{i+1}$ from $\po f_i$ by substituting variable $x_{n-i}$ by some constant $c \in Z_q$. We choose value $c$ to minimize $|\po f_{i+1}^{-1}(y)|$, but to keep $|\po f_{i+1}^{-1}(y)| > 0$. Note that as there are at least two possible values for $c$ preserving $|\po f_{i+1}^{-1}(y)| > 0$, namely $a_{n-1}$ and $b_{n-1}$ so $1\leq|\po f_{i+1}^{-1}(y)|\leq\frac{|\po f_i^{-1}(y)|}{2}$. Moreover, it is easy to see that $\deg \po f_i\geq \deg \po f_{i+1}$.

Case $|\po f_i^{-1}(y)|\geq q^q$ is a bit more complicated since we want to reduce the size of $\po f^{-1}(y)$ faster than in the previous case. As $|\po f_i^{-1}(y)|\geq q^q$ we can find $q$ elements of $\po f_i^{-1}(y)$, say $v^1$, $v^2$,\ldots, $v^{q}$, which treated as a vectors over field $\m F_q$ are linearly independent. Hence there exists $(0,\ldots,0)\not=(\beta_1,\ldots,\beta_{n-i})\in F_q^{n-i}$ such that for every $a\in F_q$  there exists $k$ such that
\[
\sum_{j=1}^{n-i}\beta_j\cdot v^k_j=a.
\]
Since, $v^j$'s are taken from $\po f_i^{-1}(y)$ it follows that for every $a\in F_q$ the system of equations 
\[
\begin{cases}
\po f_i(\o x)=y\\
 \sum_{j=1}^{n-i}\beta_j\cdot x_j=a
\end{cases}
\]
has a solution. Denote the set of solutions of system of equations in such the form as $S_a$. Let $u$ be such that $\beta_u\not=0$. Assume without loss of generality, that $u=n-i$. We choose $b\in F_q$ which minimize the size of set $S_b$ and produce $\po f_{i+1}$ by substituting in $\po f_i$ variable $x_{n-i}$ with 
\[
\beta_{n-i}^{-1}(b-\sum_{j=1}^{n-i-1}\beta_j\cdot x_j).
\]
Note that $\sum_{a\in F_q}|S_a|=|\po f_i^{-1}(y)|$ and hence $|S_b|\leq \frac{|\po f_i^{-1}(y)|}{q}$.  Thus, $|\po f_{i+1}^{-1}(y)|\leq \frac{|\po f_i^{-1}(y)|}{ q}$.
Besides, $\deg \po f_i\geq \deg \po f_{i+1}$. It is easy to see that sequence of polynomials constructed in presented way fulfills required conditions. 

Now, we will prove that $\deg \po f\geq n-l$. There are two cases: $\po f$ is a polynomial in one variable and $|\po f_l^{-1}(y)|=1$. If $\po f$ is an univariate polynomial then $n-l=1$ and since $\po f$ is not a constant polynomial $\deg f\geq n-l$. The case when $|\po f_l^{-1}(y)|=1$ is a bit more complicated. Notice, that there is exactly one tuple  $\o a=(a_1,\ldots,a_{n-l})\in F_q^{n-l}$ such that $\po f_l(\o a)=y$. Let 
\[
 \po f'(x_1,\ldots, x_{n-l})=1-(\po f_l(x_1+a_1,\ldots, x_{n-l}+a_{n-l})-y)^{q-1}.
\]
One can easily check that $\po f'(\o x)=1$ iff $x=(0,\ldots,0)$ and otherwise it is equal zero. Obviously $\deg \po f'\leq (q-1)\deg \po f_l$. On the other hand we can express $\po f'$ in the following way: 
\[
\po f'(x_1,\ldots,x_{n-l})=\prod_{i=1}^{n-l}(1-x_i^{q-1}).
\]
Above polynomial has degree $(q-1)\cdot(n-l)$. This is the lowest possible degree as every polynomial over field $\m F_q$ has unique representation as sum of monomials modulo identities in the form $x_i^q=x_i$. Hence, $(q-1)\deg \po f_l\geq \deg \po f'\geq (q-1)(n-l)$ and in a consequence $\deg \po f\geq \deg \po f_l\geq n-l$.  

Now, we are ready to do the final calculations. Denote $K=|\po f^{-1}(y)|$. Let $l_1$ be the number of $\po f_i$'s obtained by substituting one of variables in $\po f_{i-1}$ by a constant, and $l_2=l-l_1$ i.e the number of $\po f_i$'s we get substituting one of the variables of $\po f_{i-1}$ by linear combination of other variables. It is easy to see that $l_1\leq \log_2{q^q}=q\log_2{q}$ and $l_2\leq \log_q{K}$. Summarizing 
\[
 \deg \po f\geq n-l=n-l_1-l_2\geq n-q\log_2{q}-\log_q{K}.
\]
Hence,
\[
 q^{\deg \po f}\geq q^{n-q\log_2{q}-\log_q{K}}
\]
and finally
\[
|\po f^{-1}(y)|=K=q^{\log_q{K}}\geq q^{n-\deg\po f-q\log_2{q}}.
\]
which finishes the proof of the lemma.
\section{Deterministic algorithm}\label{sec-determ} 
In this Section we prove Theory \ref{thm-determ-alg}. Let $\m A$ be a fixed supernilpotent algebra of prime power order $q^h$ and 
\begin{equation}\label{eq-determ}
\po p(\o x)= \po g(\o x)
\end{equation}
be a given equation over $\m A$. By Lemma \ref{lm-eqA-to-eqF} there exists polynomial $\po f$ over $\m F_q$ of degree $d=|A|^{\log_q{m}+1}$ and arity $hn$, where $m$ is bound on arity of basic operation of $\m A$, such that $\po f(F_q^{hn})\subseteq\set{0,1}$ and an equation 
\begin{equation}\label{eq-f}
\po f(x_1^1,\ldots,x_1^h,\ldots,x_n^1,\ldots,x_n^h)=1
\end{equation}
has a solution iff equation \eqref{eq-determ} has a solution. We have even more, $\o a\in A^{n}$  is a solution of equation \eqref{eq-determ} iff $\po f(\pi_1(a_1),\ldots, \pi_h(a_1),\ldots,\pi_1(a_n),\ldots,\pi_h(a_n))=1$. 
Thus, it is enough to show the algorithm solving equation $\po f(\o x)=1$. 

Our algorithm treats circuit as a black-box and checks the set $S_{n, h}\in F_q^{nh}$ of potential solutions of polynomial size in $n$ with such the property that if  equation \eqref{eq-f} has a solution it has solution contained in $S_{nh}$. The algorithm returns "yes" if it finds the solution in the hitting set, and "no" otherwise. In the next paragraph we will show that such the set $S_{nh}$ exists for every $n$ and it can be compute in polynomial time. If $\po f$ is a constant function  then the algorithm obviously returns proper answer for every non-empty set of potential solutions as a hitting set. Hence, we can assume that $f$ is not a constant function. 

As every polynomial over $\m F_q$ also polynomial $\po f$ can be presented as a sum of pairwise different monomials multiplied by nonzero constants  from the field. Let $t$ be a monomial taken from this presentation which contains the biggest number of different variables. From the fact that degree of $\po f$ is bounded by $d$ we have that $t$ depends on at most $d$ variables. Now, let consider the polynomial $\po f'$ formed by substituting variables not contained in  $t$ by $0\in F_q$. Note that $\po f'$ is not syntactically equal any constant and hence it is not a constant function as every polynomial function over finite filed has unique representation (modulo equations $x^q = x$ for variables). Therefore, there exists solution to the equation $\po f'(\o x)=1$. Such the solution corresponds to the solution of equation \eqref{eq-f} in which at most $d$ variables a not equal $0$. Hence, we obtain that equation \eqref{eq-f} has a solution if it has a solution in which at most $d$ variables are not equal $0$. There are $O((q^hn)^d)=O(n^d)$ valuations of variables in which at most $d$ variables are different than $0$. Thus, to check if equation \eqref{eq-f} has a solution it is enough to check $O(n^d)$ potential solutions and it can be done in time $O(n^dk)$, where $k$ is a size of circuit on the input. 

\section{Randomized algorithm}\label{sec-rand}
In this section we will prove Theorem \ref{thm-rand}  which says that there exists linear time Monte Carlo algorithm solving \csat{} for supernilpotent algebras. More precisely, we will prove that if there exists solution to the equations over fixed supernilpotent algebra of prime power order then checking random assignments of variables with uniform distribution we will find the solution with probability at least $c$ for some $c>0$. 

Let 
\[
\po p(x_1,\ldots,x_n)=\po g(x_1,\ldots,x_n)
\]
be a given equation over supernilpotnent algebra $\m A$ of prime power order $q^h$. By Lemma \ref{lm-eqA-to-eqF} we get function $\po f$ which is $h n$-ary polynomial over  $\m F_q$ such that $\o a\in A$ is a solution to above equation  iff $\po f(\pi_1 a_1,\ldots,\pi_h a_1,\ldots,\pi_1,a_n,\ldots  ,\pi_h a_n )=1$. Moreover, the degree of $\po f$  is bounded by constant $d$ which depend only on $\m A$.

Now, by Lemma \ref{lm-zippel} as $\po f$ is $nh$-ary we obtain that $|\po f^{-1}(1)|\geq q^{nh-\deg \po f-q\log_{2}{q}}\geq q^{nh-d-q\log_{2}{q}}$. Observe that $\frac{|\po f^{-1}(1)|}{|A|^n}$ the fraction of assignments of variables for which $\po f$ is equal $1$ is at least $c=\frac{q^{nh-d-g\log_{2}{q}}}{q^{nh}}=q^{-d-q\log_2{q}}$. This bound does not depend on $\po f$ and $n$. Hence, linear time randomized algorithm which picks the assignments of variables with uniform distribution and check if picked assignments is a solution to the equation is a  $c$-correct true-biased Monte Carlo algorithm solving $\csat{\m A}$.

\section{Conclusions}\label{sec-conc}

The main idea of presented in this paper deterministic black-box algorithm  correctness proof is translating polynomials of nilpotent algebra $\m A$ of prime power order to polynomial over $\m F_q$ of small degree $d\leq |A|^{\log_q m +1}$. This allowed us to create the hitting sets for $\csat{\m A}$ by translating hitting sets  for  bounded degree polynomial equations over $\m F_q$. It is worth to emphasize that this reasoning works for $\textit{any}$ hitting set. This means that any black-box algorithm for polynomials over $\m F_q$ of degree at most $d$ translates to an algorithm solving equations over supernilpotent algebras of prime power order. As each variable from $\m A$ (in the reduction from $\csat{\m A}$ to polynomial equations) is factored to at most $\log(|A|)$ variables, the reduction does not affect the time complexity too much. If for instance we have some black-box algorithm for polynomial equation with hitting set of size $O(n^c)$, the same upper bound holds for $\csat{\m A}$.

On the other hand it's easy to prove the dual theorem. For any polynomial equation over $\m F_q$ of degree at most  $d=\frac{|A|^{\log_q{m}}}{m}$ there is nilpotent algebra $\m A$ of size $q^h$ and maximal arity of basic operation $m$ such that any black-box algorithm for the algebra $\m A$ translates to black box algorithm for solving equations over $\m F_q$ of degree at most $d$. To see it, we will consider the following example.

\begin{ex}\label{ex-nil-to-poly}
Let $\m A[h,m]=(A_h,+,p_1,\ldots, p_{h-1})$ be an algebra such that:
\begin{itemize}
    \item $(A_h,+)=\m Z_q^h$,
    \item $\pi_{i} p_i(x_1, \ldots, x_m) = \prod_{j=1}^k \pi_{i+1}  x_j$
    \item $\pi_j p_i(x_1, \ldots, x_m) = 0$ for $j\neq i$
\end{itemize}
\end{ex}

Note that by results of \cite{fm} algebra $\m A$ from Example \ref{ex-nil-to-poly} is supernilpotent and belongs to congruence modular variety. It easy to see that every equation between polynomials over $\m F_q$ of degree bounded by $d=m^{h-1}=\frac{m^{\log_q{|A|}}}{m}=\frac{|A|^{\log_q{m}}}{m}$ can bee easily translate into equation over $\m A$.  Moreover, projections on the first coordinate of element of any hitting set for $\csat{\m A}$ is a hitting set for solving equations of polynomials over $\m F_q$ of degree bounded by $d$.

In the light of above paragraphs, to obtain efficient black-box algorithm solving \csat{} over supernilpotent algebras it's enough to produce black-box algorithm for solving bounded degree equations for polynomials over fields and translate it to black-box algorithm for supernilpotent algebras since any other black-box algorithm for supernilpotent algebras cannot be much more efficient (in terms of size of the algebra and maximal arity of operation). So it seems that the right approach to find asymptotically optimal deterministic algorithm for supenilpotent algebras is to find optimal algorithm for polynomials of bounded degree.

There is a big disproportion between computational complexity of deterministic and probabilistic algorithms presented in this paper. Hence, it would not be surprising if there was an effective derandomization of our Monte Carlo algorithm which would result in new fast deterministic algorithm solving \csat{}.  What is also worth noting is the fact, 
that there is one probabilistic algorithm for all supernilpotent algebras that is probabilistic $\textsc{FPT}$ in terms of the algebra signature. It is nontrivial result, because if we were allowed to present the signature of supernilpotent algebra on the input, such a problem would be \npc (to prove it, we can use construction of the algebra showed in Example \ref{ex-nil-to-poly} to encode $q$-coloring).

\end{document}